\documentclass[journal,twoside]{IEEEtran}
\usepackage{amsmath}
\usepackage{amssymb}
\usepackage{amsthm}
\usepackage{cancel}
\usepackage{cite}
\usepackage{epsf}
\usepackage{subfigure}
\usepackage{epsfig,graphics}
\usepackage{multicol}
\newcommand{\RR}{\mathbb{R}}

\DeclareMathOperator{\tr}{tr}
\DeclareMathOperator{\prob}{\mathbb{P}}
\DeclareMathOperator{\ex}{\mathbb{E}}
\DeclareMathOperator{\var}{var}
\DeclareMathOperator{\cov}{cov}

\newtheorem{proposition}{Proposition}
\newtheorem{lemma}{Lemma}
\newtheorem{corollary}{Corollary}

\theoremstyle{definition}

\theoremstyle{remark}
\newtheorem*{remark*}{Remark}
\newtheorem{remark}{Remark}
\def\vec#1{{\bf #1}}

\newfont{\bb}{msbm10 scaled 1100}

\begin{document}
\bstctlcite{BSTcontrol}

\title{SINR Statistics \\ of Correlated MIMO Linear Receivers} %
%\title{SINR Distribution of Linear MIMO Receivers} %

\author{Aris L. Moustakas and Pavlos Kazakopoulos
\thanks{ALM (arislm@phys.uoa.gr) and PK (p.kazakopoulos@phys.uoa.gr) are with the Physics Dept., Athens Univ., 157 84  Athens, Greece. The work in this paper was presented in part in IEEE International Symposium on Information Theory, St. Petersburg, Aug. 2011}}
%\date{\today}

%\pubid{0000--0000/03\$17.00~\copyright~2003 IEEE}

\maketitle

\begin{abstract}
Linear receivers offer a low complexity option for multi-antenna communication systems. Therefore, understanding the outage behavior of the corresponding SINR is important in a fading mobile environment. In this paper we introduce a large deviations method, valid nominally for a large number $M$ of antennas, which provides the probability density of the SINR of Gaussian channel MIMO Minimum Mean Square Error (MMSE) and zero-forcing (ZF) receivers, with arbitrary transmission power profiles and in the presence of receiver antenna correlations. This approach extends the Gaussian approximation of the SINR, valid for large $M$ asymptotically close to the center of the distribution, obtaining the non-Gaussian tails of the distribution. Our methodology allows us to calculate the SINR distribution to next-to-leading order ($O(1/M)$) and showcase the deviations from approximations that have appeared in the literature (e.g. the Gaussian or the generalized Gamma distribution). We also analytically evaluate the outage probability, as well as the uncoded bit-error-rate. We find that our approximation is quite accurate even for the smallest antenna arrays ($2\times 2$).
%%%For the sake of comparison we also derive the exact distributions of the SINR, with and without transmission antenna correlations.
\end{abstract}


%%%%%%%%%%%%%%%%%%%%%%%%%%%
\section{Introduction}
\label{Introduction}
%%%%%%%%%%%%%%%%%%%%%%%%%%%

%\begin{itemize}
%\item This work:1: SNR special because is not a trace of matrix but just a matrix element. method can be generalized for {\em any}, also off-diagonal matrix element of a Wishart matrix.
%\item Interesting generalization is the joint distribution of two snr's. Away from the edges we can get an answer by integrating over $u$ and $v$ over all space. We should get estimates for edges etc.
%\end{itemize}

Multi-antenna systems have been known \cite{Foschini1998_BLAST1, Telatar1995_BLAST1} to offer considerable advantages, not only at the link-level, providing higher multiplexing gains and increased robustness through diversity, but also at a system-level by allowing a more effective interference mitigation in a multi-user setting. It is therefore no surprise that next generation wireless communications systems will include multi-antenna systems \cite{Standard802.11n} in order to capitalize on these benefits. To obtain the full advantages from multiple antennas, it is necessary to have an optimal receiver structure, which however is quite complex to implement in real systems. Instead, low complexity, albeit suboptimal, linear receivers offer a practical alternative.

Such receivers include the so-called MMSE (minimum mean square error) and the zero-forcing (ZF) receivers, as well as a new class of receivers recently proposed \cite{Muller2001_LowComplexityLinearReceivers} called moment-based receivers. In addition to the simplification due to the linearization of the received signal operation, the received signal may then be iteratively treated to cancel the interference from other antennas. However, in many cases, even this may impose significant complexity. An even simpler receiver structure can be constructed, in which, after the linear spatial equalization the data is decoded in a single-input single-output fashion\cite{Wolniansky1998_VBLAST_Indoors1, Caire2003_LowComplexitySTCoding}. Here we will focus on the latter, especially since we are interested in the cellular context, with separated transmitter antenna arrays in the uplink with a multi-antenna receiver terminal.

The throughput performance depends on the ability of the linear receiver structure to mitigate interference. One very useful method to quantify the performance is through the asymptotic analysis of the signal to interference and noise ratio (SINR) for the receiver in the limit of large antenna numbers using tools from random matrix theory. Its application was initially spearheaded in the context of Direct-Sequence Code-Division-Multiple-Access (DS-CDMA) where the effective channel consists of the matrix of pseudorandom codes. In this direction, the first breakthrough was made by \cite{Tse1999_LinearMultiuserReceiversEffectiveInterferenceEtc, Verdu1999_MIMO1}, who showed that in the infinite matrix-size limit the SINR of a fixed random channel realization converges to its mean. Later, similar results were obtained for more general channels \cite{Peacock2006_UnifiedLinearReceiversRMT, Chaufray2008_CDMA_MMSE_RMT}. More recently, the effectiveness of linear receivers were analyzed in terms of the total throughput from all transmitting nodes in the asymptotic limit \cite{Artigue2010_PrecoderDesignMIMO_MMSE_RMT, McKay2010_MI_MMSE_RMT, Kumar2009_LinearDMT}.

Nevertheless, one often needs to assume that the fading channel is ``quasi-static'', i.e. varies in time much more slowly than the typical coding delay. In this case the channel matrix and hence the SINR have to be considered as random quantities. In this regime, the relevant performance metric is the ``rate (or SINR) versus outage probability'' tradeoff \cite{Biglieri1998_FadingChannels}, captured by the cumulative distribution function of the SINR. This situation is especially relevant in the context of multi-antenna channels, when the number of antennas is usually much smaller than the size of the CDMA codes.

In a seminal work \cite{Tse2000_MMSEFluctuations} the authors proved the asymptotic normality of the SINR for the MMSE and ZF receivers when all transmitters have equal power. The normality of the SINR was later extended to the normality of the multiple access interference (MAI) of CDMA channels \cite{Zhang2001_Gaussian_forMAI_MMSE_DS_CDMA} and a variety of linear receivers \cite{Guo2002_AsymptoticNormalityMMMSE_ZF}. More recently, \cite{Liang2007_MMSEAsymptotics, Kammoun2009_CLT_MMSE_RMT} showed the normality of the MMSE SINR, including the case of the mismatched receiver. Interestingly, \cite{Liang2007_MMSEAsymptotics} showed also that the logarithm of the SINR becomes asymptotically normal. Unfortunately, and in contrast to the total mutual information, the Gaussian approximation for the SINR is extremely inaccurate, unless the number of antennas is quite large. As a result, inspired by the fact that the SINR for the equal power MIMO ZF receiver has a Gamma distribution\cite{Tse2000_MMSEFluctuations, Gore2002_MIMO_ZFReceiver}, several works were devoted to approximating the SINR statistics with other distributions, such as the Beta distribution for the SINR of the CDMA ZF receiver \cite{Muller1997_CDMA_CapacityLinearReceivers, Schramm1999_SpectralEfficiencyMMSEReceivers}, or the Gamma and generalized Gamma distributions \cite{Li2006_MIMO_MMSE_SINR_Distribution, Kammoun2009_BER_Outage_Approximations_MMSE_MIMO, Armada2009_BitLoadingMIMO, Li2011_BER_MIMO_MMSE}, in which case their first three moments were fitted to match the actual SINR distribution. Nevertheless, this methodology, although perhaps providing good agreement under certain conditions, is ad-hoc and does not offer any intuition on the SINR statistics. The same can be argued for the calculation of the exact probability density function (PDF) and the cumulative density function (CDF) of the MMSE SINR \cite{Kiessling2003_ExactMMSE_SINR} using ratios of determinants, a method however which is only valid for uncorrelated channels at the receiver.

In this paper, we take a different approach. Instead of trying to prove Gaussian behavior close to the peak of the distribution of SINR, we develop a large-deviations methodology, which allows us to calculate the distribution of the SINR arbitrarily far from its most probable, mean value. The success of our method lies on the fact that we can exactly express the moment generating function (MGF) of the SINR as the moment generating function of the difference of two correlated MIMO mutual information functions. Taking advantage of the robustness of the Gaussian approximation of the MIMO mutual information we obtained an expression of the MGF of the SINR correct to $O(1/M)$. We are then able to obtain the full distribution of the SINR for both MMSE and ZF with similar precision. It is therefore no surprise that our results are very close to the exact ones down to the smallest MIMO systems ($2\times 2$). It is worth mentioning a related recent work \cite{Moustakas2011_SINR_MMSE_PhysPolonicaB} in which we used the Coulomb Gas method \cite{Majumdar2006_LesHouches} to calculate the leading term $O(M)$ in the exponent of the SINR for uncorrelated channels. In that work we demonstrate that the large deviations tails are determined by the behavior of a single singular value of the channel matrix.

{\em Outline:}
In the next section we present the channel model and introduce the MMSE and the ZF SINR. In Section \ref{sec:technical} we present our analytical results, providing the PDF, CDF and BER for both MMSE and ZF SINRs. In Section \ref{sec:numerical simulations} we demonstrate their validity numerically and we conclude in Section \ref{sec:conclusion}. Appendices \ref{app:proof_gamma_vs_DI}, \ref{app:proof_pdf}, \ref{app:proof_cdf} and \ref{app:proof_ber} contain details on the proofs of Lemma \ref{lemma:gamma_vs_DI},  Propositions \ref{lemma:pdf}, \ref{lemma:cdf} and \ref{lemma:ber}, respectively.

%%%%%%%%%%%%%%%%%%%%%%%%%%%%%%
\section{Channel Model}
\label{sec:channel model}
%%%%%%%%%%%%%%%%%%%%%%%%%%%%%%

In this section we define the channel model. The receiver array has $M$ antennas, receiving the signal from $K+1$ transmitter arrays, not necessarily collocated. Without loss of generality\footnote{For example, we may assume that the zeroth and the first arrays correspond to the same transmitter} we assume that the 0th transmitter has a single antenna, while the remaining $K$ transmitters have $N_k$ antennas each for $k=1,\ldots,K$. The $M$-dimensional received signal vector $\vec y$ can be written as
\begin{equation}\label{eq:channel_def}
    \vec y = \vec R_0^{1/2}\vec g_0 x_0 + \sum_{k=1}^{K} \vec R_k^{1/2}\vec G_k \vec x_k + \vec z
\end{equation}
In the above equation $\vec z$ is the  noise vector, with complex Gaussian elements $\sim {\cal CN}(0,1)$. The transmitted signal amplitudes $x_0$ and $\vec x_k$ have i.i.d. elements with variance $p_0$ and $p_k \vec I_{N_k}$ respectively, where $p_k$ are the average transmitted power per antenna from the $k$th array with $k=1,\ldots,K$. The channel vector from transmitter 0 is $\vec R_0^{1/2}\vec g_0$, where $\vec g_0$ is an $M$-dimensional vector with i.i.d. entries $\sim {\cal CN}(0,1)$. $\vec R_0$ is the $M$-dimensional receive-side correlation matrix of the channel originating from user 0, normalized so that $\tr \vec R_0=M$. Similarly, the channel matrix from the $k$th user is $\vec R_k^{1/2}\vec G_k$, where $\vec G_k$ is a $M\times N_k$ matrix with i.i.d. elements  $\sim {\cal CN}(0,1/N_k)$ and $\vec R_k$ has the same interpretation and properties as $\vec R_0$. To be concrete, we will assume that all correlation matrices $\vec R_k$, for $k=1,\ldots,K$ are positive semidefinite, while $\vec R_0$ is positive definite. Also, we assume that their eigenvalue spectra converge to proper probability distributions for large $M$. We will be interested in calculating the SINR of transmitter $0$ in the presence of the other transmitters and noise. For notational convenience we also define the matrix $\vec H_0 = [\vec R_1^{1/2} \vec g_1 \sqrt{p_1},\ldots,\vec R_{K}^{1/2} \vec G_{K} \sqrt{p_{K}}]$ as well as the matrix $\vec H=[\vec R_0^{1/2} \vec g_0 \sqrt{p_0}, \vec H_0]$.

This channel model describes a set of transmitting antennas dispersed in a cellular setting with their signal arriving possibly from different mean angles and/or with different angle-spreads at the receiver array, thereby having different receive correlation matrices. Of course, not all correlation matrices need to be different, e.g. if some of the interfering antennas are collocated. To obtain analytic results we will take the limit of large $M$ and $N_k$ ($k=1,\dots,K$), with the ratios
\begin{equation}\label{eq:n_k_def}
    n_k = \frac{N_k}{M},
\end{equation}
as well as the number of arrays $K$ fixed in that limit. In the remainder of the paper the term  ``large $M$ limit'' will denote both $N_k$ and $M$ going to infinity, while keeping the corresponding ratios $n_k$ constant and finite. For notational convenience, we define $N_{tot}=\sum_{k=0}^K N_k$. Despite the assumptions above, we will apply and test our results in the case when $M$ and $K$ are not too large and $N_k=1$.

\subsection{MMSE Receiver}
\label{sec:MMSE_receiver}

The SINR of the 0-th MMSE transmitter above can be expressed as
\begin{eqnarray}\label{eq:SINR_mmse_def}
    \gamma(\vec H) &=& \frac{p_0}{M}\vec g_0^\dagger \vec R_0^{1/2}\vec L^{-1}\vec R_0^{1/2} \vec g_0 \\
    \vec L &=& \vec I_M + \vec H_0 \vec H_0^\dagger
\end{eqnarray}
with the second line serving as the definition of $\vec L$. Our objective is to evaluate the probability density function of $\gamma(\vec H_0)$, omitting the $\vec H_0$ dependence when obvious.

\subsection{ZF Receiver}
\label{sec:ZF Receiver}

The SINR of the zero-forcing (ZF) receiver can be obtained in a similar fashion. In this case, we focus only in the case $M\geq N_{tot}$. Then the SINR for this receiver can be expressed as a limit of the standard MMSE SINR (\ref{eq:SINR_mmse_def}) as follows
\begin{eqnarray}\label{eq:SINR_zf_def}
    \gamma(\vec H) &=& \frac{p_0}{M}\frac{1}{\left[\left\{\vec H^\dagger\vec H\right\}^{-1}\right]_{11}}
    \\ \nonumber
    &=& \frac{p_0}{M} \lim_{z\rightarrow 0^+} z\left(\frac{1}{\left\{\left[\vec I_{N} + z^{-1}\vec H^\dagger \vec H\right]^{-1}\right\}_{11}}-1\right)
    \\ \nonumber
    &=& \frac{p_0}{M} \lim_{z\rightarrow 0^+}  \vec g_0^\dagger \left[\vec I_{M}+z^{-1}\vec H_0\vec H_0^\dagger \right]^{-1}\vec g_0
    \\ \nonumber
    &=& \frac{p_0}{M} \lim_{z\rightarrow 0^+}  z\vec g_0^\dagger \left[z\vec I_{M} +\vec H_0\vec H_0^\dagger \right]^{-1}\vec g_0
\end{eqnarray}
The inverse of the matrix in the right-hand side of the first equality is finite only for $M\geq N_{tot}$ with probability one. The second equality results from taking the $z\rightarrow 0^+$ limit. The third equality above results from the matrix inversion lemma \cite{Verdu_MUD_book}. Following the same argumentation as in Section \ref{sec:MMSE_receiver} we obtain the moment generating function as in (\ref{eq:PDF_Fourier}), with $\vec H_0\vec H_0^\dagger$ in (\ref{eq:DeltaI_def}) replaced by  $z^{-1}\vec H_0\vec H_0^\dagger$. The expression in the fourth line, easily derived from the third, showcases the singular nature of the $z\rightarrow 0$ limit, which focuses on the projection of the kernel of $\vec H_0\vec H_0^\dagger$ to  $\vec g_0$, which, for $M\geq N_{tot}$ is guaranteed to be non-empty. For compactness, below we will continue to use this dummy variable $z$, setting it equal to $z=1$ and $z=0^+$ for the cases of the MMSE and ZF SINR, respectively.

%%%%%%%%%%%%%%%%%%%%%%%%%%%%%%
\section{Results}
\label{sec:technical}
%%%%%%%%%%%%%%%%%%%%%%%%%%%%%%

In this section we will go through the basic steps of the calculation of the probability distribution (PDF), the outage distribution (CDF) and the BER of the SINR denoted by $\gamma$.  We start with a very useful first result for the moment generating function of $\gamma$.

\begin{lemma}[MGF of $\gamma$]\label{lemma:gamma_vs_DI}
The moment generating function of $\gamma$ for the MMSE (\ref{eq:SINR_mmse_def}) and the ZF case (\ref{eq:SINR_zf_def}) can be written in the following form
\begin{eqnarray}\label{eq:PDF_Fourier}
    g_M(s) &=& \ex_{\vec H_0}\left\{ e^{-\Delta I\left(s, \vec H_0 \right)}\right\}
\end{eqnarray}
where  $\Delta I(s,\vec H_0)$ is given by
\begin{eqnarray}\label{eq:DeltaI_def}
\Delta I(s,\vec H_0) &=& \ln\det\left[z(\vec I_M + s p_0 \vec R_0) + \vec H_0\vec H_0^\dagger\right] \nonumber \\
                      &-& \ln\det\left[z\vec I_M +  \vec H_0\vec H_0^\dagger\right]
\end{eqnarray}
The parameter $z$ takes the $z=1$ for the MMSE SINR (\ref{eq:SINR_mmse_def}) and the limiting value $z=0^+$ for the ZF SINR (\ref{eq:SINR_zf_def}).
\end{lemma}

\IEEEproof{
See Appendix \ref{app:proof_gamma_vs_DI}.
}

\begin{remark}
Once again we see that the limit $z\rightarrow 0^+$ in (\ref{eq:PDF_Fourier}) is not trivial because the matrix $\vec H_0 \vec H_0^\dagger$ has a non-empty kernel.
\end{remark}

\begin{remark}
The usefulness of this result is that it makes the connection of the moment generating function of the SINR to a difference of mutual information functions for the remaining $K$ users. This will allow us to take advantage of the Gaussian behavior of this difference of mutual informations \cite{Moustakas2003_MIMO1, Hachem2006_GaussianCapacityKroneckerProduct} close to their ergodic values, in order to analyze the large deviations of the distribution of $\gamma$ arbitrarily far away from {\em its} ergodic value. Note that the above argument holds for general $\vec H_0$, as long as the logdets difference above remains Gaussian, as e.g. in \cite{Hachem2008_Capacity_INDCorrelation}.
\end{remark}

\begin{figure*}%[htbp]
\begin{center}
\subfigure[PDF of MMSE SINR (dB)]
{\includegraphics[width=0.49 \textwidth]{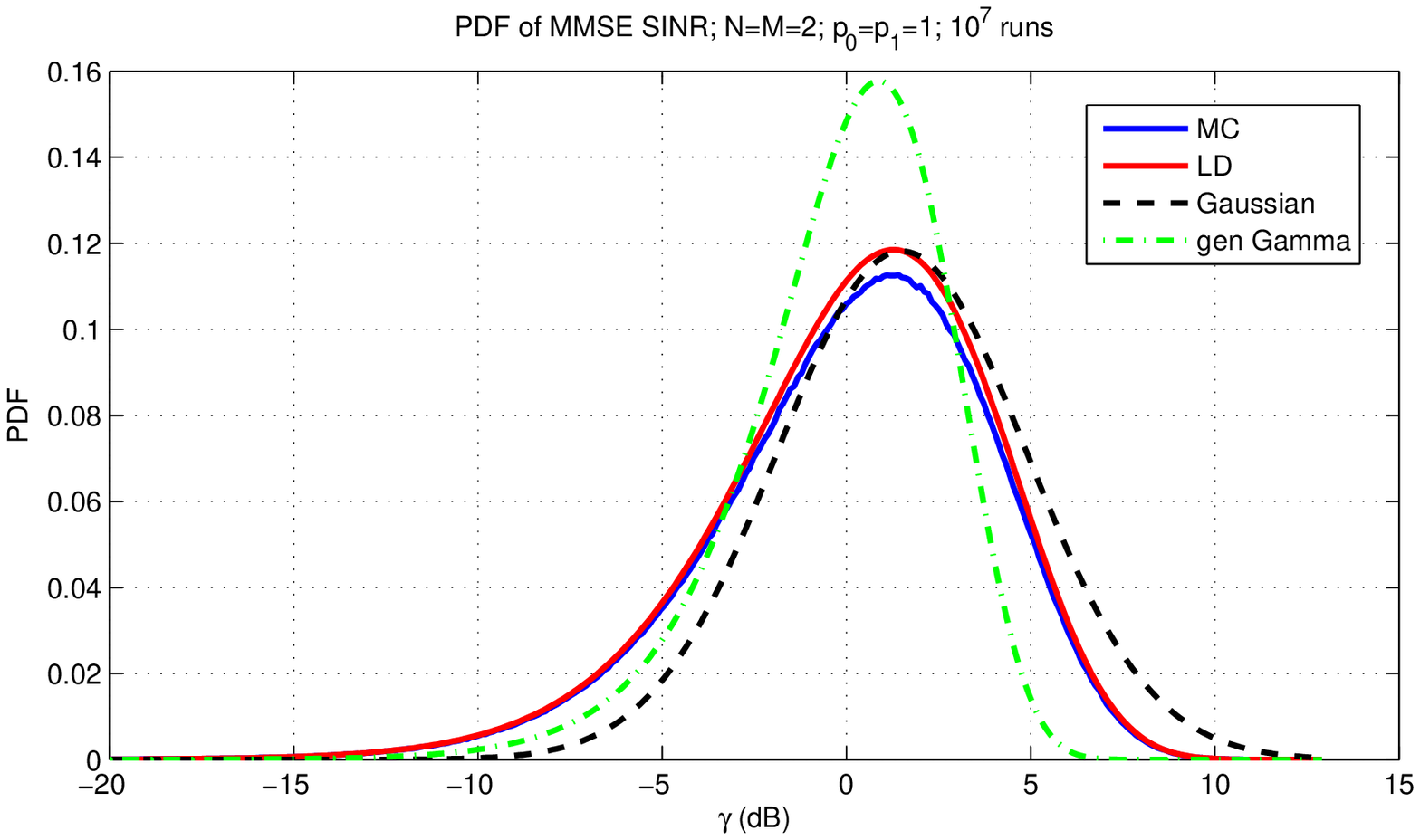}}
\subfigure[CDF of MMSE SINR (dB)]
{\includegraphics[width=0.49 \textwidth]{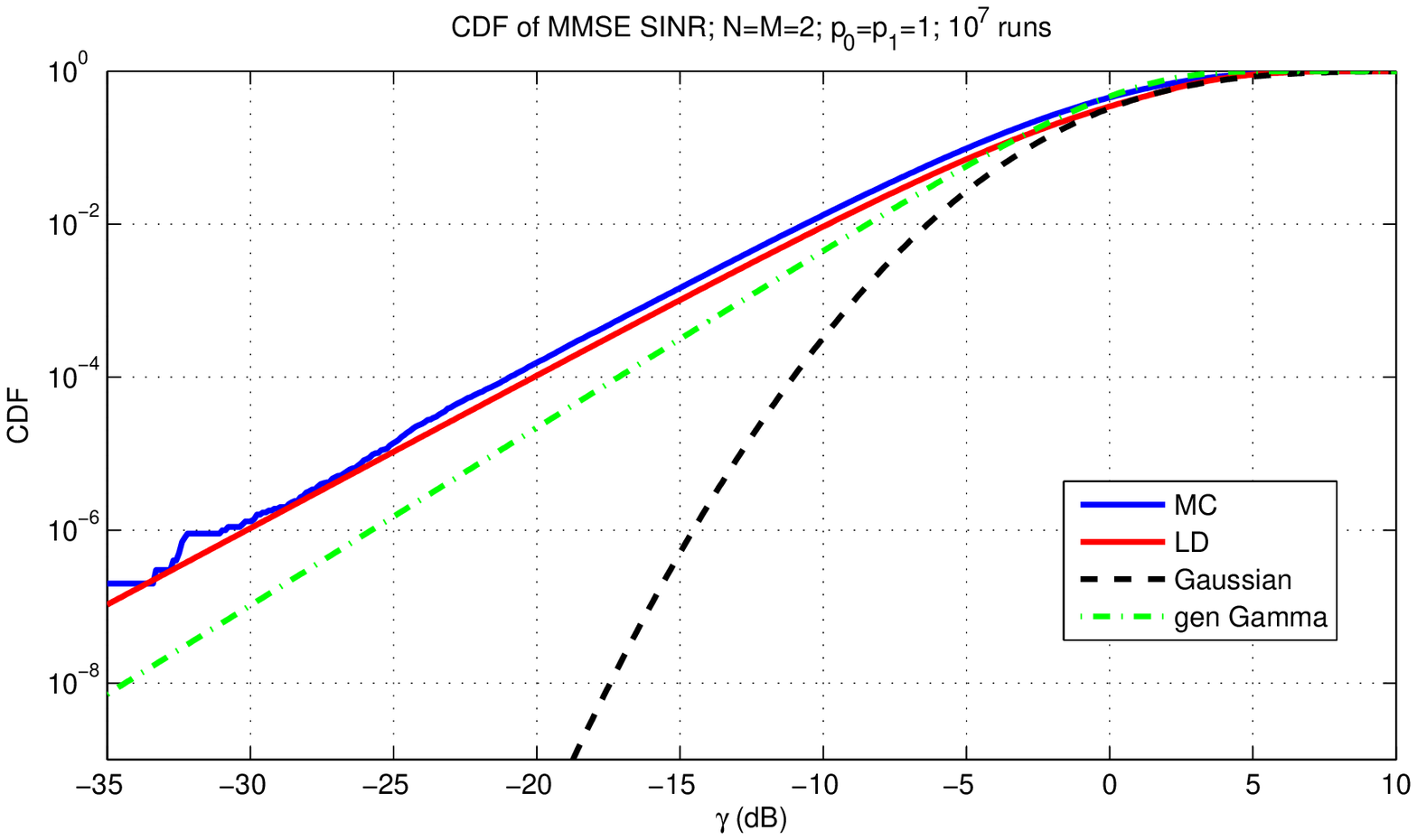}}
\caption{(a) Probability density (PDF) curves and (b) Outage probability (CDF) curves for the MMSE SINR in dB for $N=M=2$, with transmission and interference powers $p_0=2$ and $p_1=1$ respectively. The channel elements are assumed to be iid. We plot the PDF for the Monte Carlo-generated simulations (MC), the LD approximation, the Gaussian approximation and the generalized Gamma approximations. In the Gaussian curves we have used  ${\cal N}(\gamma_{dB, erg}, \sigma^2_{dB})$ with $\gamma_{dB,erg}=10*\log_{10} \gamma_{erg}$ and $\sigma^2_{dB}=\sigma^2_{erg}/( \gamma_{erg}\ln 10)^2$. The ergodic mean $\gamma_{erg}$ and variance $\sigma^2_{erg}$ of the SINR can be calculated directly, see e.g. \cite{Kammoun2009_BER_Outage_Approximations_MMSE_MIMO}. The generalized gamma curves have been plotted using the parameters of the generalized gamma distribution as calculated in \cite{Li2006_MIMO_MMSE_SINR_Distribution}.}
\end{center}
\label{fig:MMSE_curves}
\end{figure*}

\subsection{Derivation of PDF}
\label{sec:PDF}

We will now obtain the probability distribution density of the SINR. This density may be expressed as an expectation of a Dirac $\delta$-function as follows:
\begin{equation}\label{eq:PDF_definition_delta_fn}
    P_M(\gamma) = \ex_{\vec H}\left[ \delta\left(\gamma - \frac{z p_0}{M} \vec g_0^\dagger \left[z\vec I_M + \vec H_0 \vec H_0^\dagger\right]^{-1}\vec g_0\right) \right]
\end{equation}
The parameter $z$ will take the value of $z=1$ for the case of the MMSE SINR introduced in Section \ref{sec:MMSE_receiver}, while, as discussed in Section \ref{sec:ZF Receiver}, the $z=0^+$ limit will correspond to the ZF SINR. The following proposition provides an analytic expression of the probability density of the SINR, valid for all $\gamma>0$ in the large $M$ limit.
\begin{proposition}[PDF of SINR]\label{lemma:pdf}
Let $\overline{P}_M(\gamma)$ be given by
\begin{eqnarray}\label{eq:PDF_sinr1}
    {\overline P}_M(\gamma) &=& \frac{1}{\sqrt{2\pi}} e^{Ms_0\gamma-I_{erg}(s_0)+I_{erg}(0)+\frac{v_1(s_0)+v_2(s_0)}{2}}
\end{eqnarray}
In the above equation, $I_{erg}(s)$ is the ergodic mutual information given by
\begin{eqnarray}\label{eq:Ierg_gen_def}
I_{erg}(s) &=& \tr \ln\left[\vec I_M + \vec R_0 p_0 s+ \sum_{k=1}^{K} \vec R_k r_k(s)\right] \\ \nonumber
&+& \sum_{k=1}^{K} N_k \ln(z + p_k t_k(s)) - \sum_{k=1}^{K} N_k r_k(s) t_k(s) \\
r_k(s) &=& \frac{p_k}{z + p_k t_k(s)} \label{eq:rs_gen_def} \\
t_k(s) &=& \frac{1}{N_k}\tr \left[\vec R_k\left(\vec I_M + p_0 s \vec R_0+\sum_{q=1}^{K} \vec R_q r_q(s) \right)^{-1} \right] \label{eq:ts_gen_def}
\end{eqnarray}
for $k=1,\dots,K$. The parameter  $z$ takes the value $z=1$ for the case of the MMSE SINR and the value $z=0$ for the case of the ZF SINR. The variable $s_0$ in (\ref{eq:PDF_sinr1}) is evaluated through the saddle-point equation
\begin{eqnarray}\label{eq:gamma_solution}
\gamma &=& \frac{1}{M}I_{erg}'(s_0) \\ \nonumber
&=&p_0 \frac{1}{M} \tr \left[\vec R_0\left(\vec I_M + p_0 s_0 \vec R_0+\sum_{k=1}^{K} \vec R_k r_k(s_0) \right)^{-1} \right]
\end{eqnarray}
$I_{erg}'(s)$ is the derivative of $I_{erg}(s)$ with respect to $s$. The expressions of the $O(1)$ terms $v_1(s_0)$ and $v_2(s_0)$ are given in Appendix \ref{app:proof_pdf}.
$I_{erg}(0)$ is obtained by setting $s=0$ in $I_{erg}(s)$, $t_k(s)$, $r_k(s)$ above.

Then for every $\gamma>0$, the probability density converges weakly to $\overline{P}_M(\gamma)$ in the sense that
\begin{equation}\label{eq:PDF_convergence}
    \lim_{M\rightarrow\infty} M\left|P_M(\gamma)-\overline{P}_M(\gamma)\right|<\infty
\end{equation}
\end{proposition}
\begin{proof}
See Appendix \ref{app:proof_pdf}.
\end{proof}
\begin{remark}
As it will become clear in the appendix, this result means that for large $M$ the PDF of the SINR becomes asymptotically equal with $\overline{P}_M(\gamma)$, up to corrections of $O(1/M)$.
\end{remark}
\begin{remark}
The solution of (\ref{eq:rs_gen_def}), (\ref{eq:ts_gen_def}) has been shown to be unique for the case of the MMSE SINR ($z=1$) \cite{Moustakas2003_MIMO1, Couillet2009_CapacityAnalysisMIMO, Dupuy2010_CapacityAchievingCovarianceFrequencySelectiveMIMOChannels}. To show that this is also the case for the $z\rightarrow 0^+$ limit, we observe that we can rewrite  (\ref{eq:rs_gen_def}) as $r_k(s)=(t_k(s)+z/p_k)^{-1}$. Hence we may view the ZF $z\rightarrow 0$ limit as the MMSE solution with $p_k\rightarrow\infty$ limit ($k=1,\ldots K$). Since the MMSE ($z=1$) solution for the $r_k(s)$, $t_k(s)$ is continuous with respect the values of $p_k$, we may take the MMSE $r_k(s)$, $t_k(s)$ solutions in the limit of large $p_k$, and then plug them in (\ref{eq:Ierg_gen_def}), setting also $z=0$.
\end{remark}

Also, note that the most probable value of $\gamma$ corresponds to the solution of (\ref{eq:gamma_solution}) for $s_0=0$. This involves the joint solution of (\ref{eq:ts_gen_def}), (\ref{eq:rs_gen_def}), which gives the correct value of the ergodic SINR \cite{Kammoun2009_BER_Outage_Approximations_MMSE_MIMO, Moustakas2009_VTC_MMSE_CorrMIMO_Capacity}. Expanding the leading $O(M)$ term in the exponent of the PDF (i.e. the first three terms) to second order in $(\gamma-\gamma_{erg})$ provides the Gaussian approximation of the PDF of the SINR. Furthermore, since $\overline{P}_M(\gamma)$ in (\ref{eq:PDF_sinr1}) is valid for all positive $\gamma$, not necessarily close to the ergodic value, it can provide the tails of the distribution accurately.

\begin{figure*}%[htbp]
\begin{center}
\subfigure[PDF of ZF SINR (dB)]
{\includegraphics[width=0.49 \textwidth]{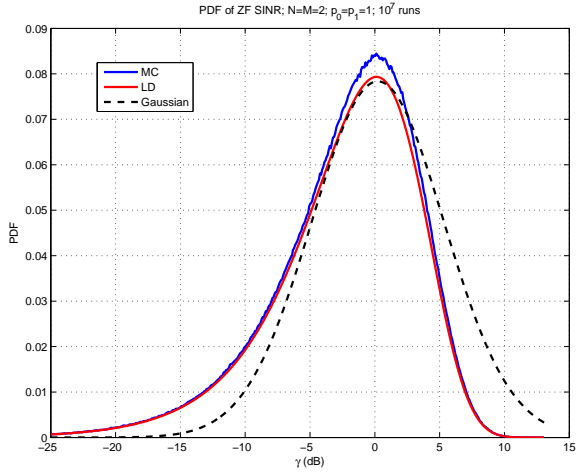}}
\subfigure[CDF of ZF SINR (dB)]
{\includegraphics[width=0.49 \textwidth]{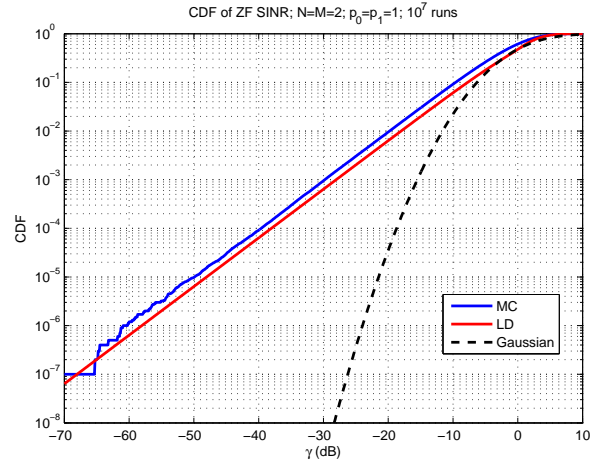}}
\caption{(a) Probability density (PDF) curves and (b) Outage probability (CDF) curves for the ZF SINR in dB for $N=M=2$, with signal and interference powers $p_0=2$, $p_1=1$, respectively. The angles of arrival of the signal and  interference paths are $\theta_0=0^o$ and $\theta = 45^o$, respectively, measured from the vertical of the receive antenna array. The angle-spreads of both paths at the receiver are $\sigma_{as}=30^o$. As in Fig. 1%\ref{fig:MMSE_curves}
, the LD curve using this approximation is consistently closer to the Monte-Carlo generated curve (MC). The way the Gaussian curve of generated is identical to Fig. 1%
\ref{fig:MMSE_curves}.}
\end{center}
\label{fig:ZF_curves}
\end{figure*}

In \cite{Moustakas2011_SINR_MMSE_ISIT} we derived a simplified expression for the case when all correlation matrices are identical. This result can also be obtained from the above analysis by setting $N_1=K$ and all other $N_k=0$, and $\vec R_0=\vec R_1 = \vec R$, while $\vec g_0\sim {\cal CN}(0,1/K)$, $\vec G_1\sim {\cal CN}(0,diag(p_1,\ldots,p_K)/K)$.
\begin{corollary}{\cite{Moustakas2011_SINR_MMSE_ISIT}}\label{lemma:pdf_special}
 Let each of the $K+1$ transmitters have a single antenna with same correlation matrix at the receiver given by $\vec R$. Then in the limit of large $K$, and $N$ with $q=K/N$ fixed the expressions (\ref{eq:Ierg_gen_def}), (\ref{eq:rs_gen_def}), (\ref{eq:ts_gen_def}) are simplified to
\begin{eqnarray}\label{eq:Ierg_special}
I_{erg}(s) &=& \tr \ln\left[\vec I_M + \vec R \left(p_0 s+ r(s)\right)\right] \\ \nonumber
&+& \sum_{k=1}^{K} \ln(z + p_k t(s)) - K r(s) t(s) \\
r(s) &=& \frac{1}{K}\sum_{k=1}^{K}\frac{p_k}{z + p_k t(s)} \label{eq:rs_def} \\
t(s) &=& \frac{1}{K}\sum_{i=1}^M\frac{R_i}{1 + R_i(p_0 s + r(s))} \label{eq:ts_def}
\end{eqnarray}
where $R_i$, ($i=1,\ldots,M$) are the eigenvalues of the matrix $\vec R$, while $z$ takes the value $z=1$ for the MMSE case and $z=0$ for the ZF case. As a result, (\ref{eq:gamma_solution}) simply becomes
\begin{eqnarray}\label{eq:gamma_solution_special}
\gamma &=& p_0 t(s_0)
\end{eqnarray}
with the corresponding expressions for $I_{erg}(0)$, $r(0)$, $t(0)$ resulting from setting $s=0$ to (\ref{eq:Ierg_special}), (\ref{eq:rs_def}), (\ref{eq:ts_def}), respectively. The expressions of $v_1(s_0)$, $v_2(s_0)$ are also accordingly simplified (see Appendix \ref{app:proof_pdf}).
\end{corollary}
To be able to compare the obtained distribution of the MMSE SINR with other proposed distributions \cite{Liang2007_MMSEAsymptotics, Kammoun2009_CLT_MMSE_RMT, Kammoun2009_BER_Outage_Approximations_MMSE_MIMO, Li2006_MIMO_MMSE_SINR_Distribution}, it is instructive to further simplify the assumptions. In particular, we have the following
\begin{corollary}
In the case of equal power transmit antennas $p_0=p_k = \rho$ and  uncorrelated receiver antennas $\vec R=\vec I_{M}$, the result simplifies and, to leading order in $N$ takes the following simple form:
\begin{eqnarray}\label{eq:PDF_P=R=I}
    \overline{P}_M(\gamma) &\propto&  e^{-K\gamma/\rho}\frac{\gamma^M}{\left(z+\gamma\right)^{K}}
\end{eqnarray}
\end{corollary}
This extremely simple result is quite remarkable. Although for large $M$ and close to the ergodic value of $\gamma$ this equation will behave approximately as a normal distribution, for general values of $\gamma$ this is far from a Gaussian or generalized Gamma-distribution. This is partly the reason why all efforts to approximate the distribution of $\gamma$ using a central limit theorem approach have largely failed, at least for relatively small values of $M$. At the same time, when $z=0$, the above distribution becomes exactly a Gamma distribution as shown in \cite{Gore2002_MIMO_ZFReceiver}.

\subsection{Outage Distribution of $\gamma$}
\label{sec:outage}

Using the expressions of the probability density $P_M(\gamma)$ from the previous section we may now evaluate the asymptotic expression of the outage probability of the SINR $P_{M,out}(\gamma_0)=\prob(\gamma(\vec H)<\gamma_0)$. It turns out that it can be evaluated using the information obtained thus far. In particular, we have
\begin{proposition}[Outage Probability]\label{lemma:cdf} Let $\overline{P}_{M,out}(\gamma)$ be given by
\begin{eqnarray}\label{eq:CDF_sinr1}
    \overline{P}_{M,out}(\gamma) &=& e^{Ms_0\gamma-I_{erg}(s_0)+I_{erg}(0)-\frac{s_0^2}{2}I''_{erg}(s_0)+\frac{v_1(s_0)}{2}} \nonumber \\
    &\cdot& Q\left(\sqrt{|I_{erg}''(s_0)|s_0^2}\right)
\end{eqnarray}
for $\gamma<\gamma_{erg}$ and
\begin{eqnarray}\label{eq:CDF_sinr2}
    1-\overline{P}_{M,out}(\gamma) &=& e^{Ms_0\gamma-I_{erg}(s_0)+I_{erg}(0)-\frac{s_0^2}{2}I''_{erg}(s_0)+\frac{v_1(s_0)}{2}} \nonumber \\
    &\cdot& Q\left(\sqrt{|I_{erg}''(s_0)|s_0^2}\right)
\end{eqnarray}
when $\gamma\geq\gamma_{erg}$. $I_{erg}''(s_0)$ is the second derivative of $I_{erg}(s)$ with respect to $s$. $Q(x)$ is defined as $Q(x)=\int_x^\infty dx e^{-x^2/2}/\sqrt{2\pi}$. The definitions of $s_0$, $I_{erg}(s_0)$ and $v_1(s_0)$ can be found in Proposition \ref{lemma:pdf} and Appendix \ref{app:proof_pdf}. The dependence of $s_0$ on $\gamma$ can be obtained through (\ref{eq:gamma_solution}). $\gamma_{erg}$ corresponds to the value of $\gamma$ in (\ref{eq:gamma_solution}) when $s=0$. The parameter $z=1$ for the MMSE (\ref{eq:SINR_mmse_def}) case and $z=0^+$ for the ZF (\ref{eq:SINR_zf_def}) case.

Then for every $\gamma>0$, the outage probability function converges to $\overline{P}_{M,out}(\gamma)$ in the sense that
\begin{equation}\label{eq:CDF_convergence}
    \lim_{M\rightarrow\infty} M\left|\prob_M(\gamma(\vec H)<\gamma)-\overline{P}_{M,out}(\gamma)\right|<\infty
\end{equation}

\end{proposition}
\begin{proof}
See Appendix \ref{app:proof_cdf}.
\end{proof}

\subsection{Evaluation of Average BER}
\label{sec:ber}

In addition to the outage probability, another important metric of performance for the linear receivers is the average uncoded bit-error probability (BER). This can be expressed as an average over $\gamma$ of  $P_e(\gamma)$, the bit-error probability conditioned on the channel realization, which for different modulations can be expressed as
\begin{eqnarray}\label{eq:MQAM_def}
    P_e(\gamma) = \left\{
                \begin{array}{cc}
                    Q(\sqrt{2\gamma}) & BPSK \\
                    Q(\sqrt{\gamma}) & QPSK \\
                    \frac{2}{\log_2 L} Q(\sqrt{\frac{3\gamma}{L-1}}) & L-QAM
                  \end{array}
                  \right.
\end{eqnarray}
where the latter expression holds approximately for large $L$ \cite{Armada2009_BitLoadingMIMO}. The average BER is given by the following
\begin{proposition}[Average BER]\label{lemma:ber}
Define the following function
\begin{eqnarray}\label{eq:lemma_ber1}
    \overline{BER}_M &=& \frac{b}{2} e^{I_{erg}(0)-I_{erg}(\frac{a}{2M})+\frac{1}{2}v_1(\frac{a}{2M})+\frac{a}{2M}I_{erg}'(\frac{a}{2M})}\\ \nonumber
    &\cdot& \Gamma\left(\frac{1}{2},\frac{a}{2M}I_{erg}'(\frac{a}{2M})\right)
\end{eqnarray}
where $\Gamma(x,y)=\int_y^\infty dt t^{x-1} e^{-t}/\Gamma(x)$ is the normalized incomplete $\Gamma$-function and $I_{erg}(s)$, $I_{erg}'(s)$ and $v_1(s)$ are the ergodic mutual information, its derivative with respect to $s$, and the variance defined in (\ref{eq:Ierg_gen_def}), (\ref{eq:gamma_solution}) and (\ref{eq:app_variance}) respectively. The above function and parameters are defined both for the MMSE ($z=1$) and the ZF $z=0$) receiver cases. Also the parameters $a,b$ describing the modulation are defined in (\ref{eq:MQAM_def}).

Then if $BER_M$ is the average uncoded bit-error rate of the MMSE (\ref{eq:SINR_mmse_def}) and the ZF (\ref{eq:SINR_zf_def}) receivers, in the limit of large $M$ we have
\begin{equation}\label{eq:BER_convergence}
    \lim_{M\rightarrow\infty} M\left|BER_M-\overline{BER}_M\right|<\infty
\end{equation}
\end{proposition}
\begin{proof}
See Appendix \ref{app:proof_ber}.
\end{proof}

\section{Numerical Simulations}
\label{sec:numerical simulations}

%\begin{corollary}
%The probability distribution of the rate $r_k=\log_2(1+\gamma_k)$, $P_r(r_k)$ can be directly obtained from above:
%\begin{equation}\label{eq:P_r}
%    P_r(r_k)=P(2^{r_k}-1) 2^{r_k} \ln2
%\end{equation}
%\end{corollary}
%We may also obtain the distribution of the SINR for the ZF receiver, by taking the appropriate $\rho\rightarrow \infty$ limit.

To test the applicability of this approach, we have performed a series of numerical simulations and have compared our large deviations (LD) approach with Monte Carlo (MC) simulations, the Gaussian approximation and the generalized gamma approximation by \cite{Li2006_MIMO_MMSE_SINR_Distribution, Armada2009_BitLoadingMIMO, Kammoun2009_BER_Outage_Approximations_MMSE_MIMO, Li2011_BER_MIMO_MMSE}. We start with the simpler case where no correlations are present in the receiver side using different powers for the transmit antennas. In Fig. 1 %\ref{fig:MMSE_curves}
we plot the probability density (PDF) and the outage probability (CDF) of the MMSE SINR in dB for the $2\times 2$ antenna case. The PDF curve of our large deviations (LD) approach is consistently closer to the Monte-Carlo (MC) numerical curves. The same is true also for the outage curves even for such small antenna arrays.

In Fig. 2 %\ref{fig:ZF_curves}
we plot the PDF and CDF curves for the zero-forcing (ZF) SINR in dB for the $2\times 2$ antenna case, using different correlation matrices $\vec R$ for the two transmitter paths. In particular, we parameterize the correlation matrix elements using the mean angle of arrival $\theta$, as measured from the vertical of the antenna array, and a Gaussian angle-spread $\sigma_{as}$ as follows:
\begin{equation}\label{eq:R_ij}
    R_{ab} = C\int_{-\pi}^{\pi} d\phi e^{2\pi i d_{ab}\sin(\phi)/\lambda} e^{-\frac{(\phi-\theta)^2}{2\sigma_{as}^2}}
\end{equation}
where $\lambda$ is the carrier wavelength, $d_{ab}$ is the distance between antennas $a,b$, taken to be $d_{ab}=(a-b)\lambda/2$ and $C$ a normalization to ensure $R_{aa}=1$. Using the above notation, the angles of arrival of the signal and the interferer are $\theta_0=0^o$ and $\theta= 45^o$, respectively, while all angle spreads are taken to be $\sigma_{as}=30^o$. In this case, we also see very good agreement with the Monte-Carlo curves.

Finally, in Fig. 3 %\ref{fig:BER_curves}
we test our predictions of the uncoded BER, both for MMSE and ZF. In Fig. 3(a) %\ref{fig:BER_curves}
we take uncorrelated receivers and compare to Monte-Carlo simulations and the generalized gamma approximation. We see that at large SNRs, the generalized gamma distribution deviates up to several dB. In contrast, our LD approximation is quite close to the numerical curve. We see similar behavior for our approximation in the ZF case. In Fig 3(b) we plot the BER as a function of angle-of arrival of the signal path, in the presence of two interfering paths, for several angle-spreads and receive array sizes. We find that low angle-spreads lead to deterioration of the BER when the signal path has the same direction of arrival as the interfering paths. In addition, we find that lower angle-spreads increase the BER away from the interferers' direction. This last observation is due to the fact that higher angle-spread leads to higher diversity and hence reduced outage probability. Interestingly, an angle-spread of just $\sigma_{as}=5^o$ is enough to make two interference paths separated by $30^o$ practically indistinguishable for $M=6, N=3$.

\section{Conclusion}
\label{sec:conclusion}

In this paper we have used a large deviation approach to calculate the key statistics of the SINR, i.e. PDF, outage probability and BER for the MMSE and ZF receivers of the Gaussian MIMO channel with arbitrary receive antenna correlations. Our results agree very well with simulations both close to the peak of the distribution as well as at its tails, where other suggested approximations, such as the Gaussian or the generalized Gamma distributions are inaccurate. As a technical byproduct, we have found an exact relationship between the SINR distribution and the moment generating function of a difference of related mutual informations. Remarkably, the accuracy of the calculated distribution, even at its tails, is a by-product of the robustness of the Gaussian behavior of the MIMO mutual information. Several direct generalizations are possible. This approach may be generalized to include multi-tap or frequency selective MIMO channels\cite{Moustakas2007_MIMO1}.

\appendices

\section{Proof of Lemma \ref{lemma:gamma_vs_DI}}
\label{app:proof_gamma_vs_DI}

The moment generating function of $\gamma$ is
\begin{equation}\label{eq:MGF_def}
    g_M(s)= \ex_{\vec H} \left [e^{-s p_0 \vec g_0^\dagger \vec R_0^{1/2}\vec L^{-1} \vec R_0^{1/2} \vec g_0}\right]
\end{equation}
We can integrate over $\vec g_0$ to obtain
\begin{eqnarray}\label{eq:app_MGF1}
   g_M(s) &=& \ex_{\vec H_0}\left\{ \int d\vec g_0 e^{-\vec g_0^\dagger\left(\vec I_M+s p_0 \vec R_0^{1/2}\vec L^{-1}\vec R_0^{1/2}\right) \vec g_0}\right\} \nonumber \\
    &=& \ex_{\vec H_0}\left\{ e^{-\ln\det\left[\vec I_M + p_0 s \vec R_0^{1/2}\vec L^{-1}\vec R_0^{1/2} \right]}\right\} \\ \nonumber
    &=& \ex_{\vec H_0}\left\{ e^{-\Delta I\left(s, \vec H_0\right)}\right\}
\end{eqnarray}
where the quantity $\Delta I(s, \vec H_0)$ is exactly (\ref{eq:DeltaI_def}). $\Delta I(s, \vec H_0)$ and therefore $g_M(s)$ will be analytic in $s$ when $ \Re(s)>-\lambda_{min}(\vec H_0)$, where $\lambda_{min}(\vec H_0)$ is the minimum eigenvalue of the matrix $p_0^{-1}\vec R_0^{-1/2}\left(\vec I_M + \vec H_0\vec H_0^\dagger\right)\vec R_0^{-1/2}$. We will assume that in the large $M$ limit $\lambda_{min}(\vec H_0)$ will converge with probability one to a fixed value $\lambda_{min}^*$ and hence for $\Re(s)>-\lambda_{min}^*$, $g_{M\rightarrow \infty}(s)$ is analytic. This has been shown for $K=1$ \cite{Paul2009_NoEigenvaluesOutsideSupport} and is expected to be true for general $K\geq 1$ \cite{Couillet2009_CapacityAnalysisMIMO}.

\section{Proof of Proposition \ref{lemma:pdf}}
\label{app:proof_pdf}

Before discussing some elements of the proof, we introduce the normalized mutual information difference as
\begin{equation}\label{eq:deltaI(s,H0_def)}
    \delta I(s,\vec H_0) = \frac{1}{M} \Delta I(s,\vec H_0)
\end{equation}
where $\Delta I(s,\vec H_0)$ is given  by (\ref{eq:DeltaI_def}). We also introduce an important property of $\delta I(s,\vec H_0)$.
\begin{lemma}[Hardening of $\delta I(s,\vec H_0)$]
\label{lem:HardeningDI}
In the limit $M\rightarrow \infty$ the quantity $\delta I(s,\vec H_0)$ converges with high probability to
\begin{eqnarray}\label{eq:Ierg(s)_gen_def_app}
\delta I_{erg}(s) = \frac{\Delta I_{erg}(s)}{M} = \frac{I_{erg}(s)-I_{erg}(0)}{M}
\end{eqnarray}
where $I_{erg}(s)$ is defined in (\ref{eq:Ierg_gen_def}), (\ref{eq:rs_gen_def}), (\ref{eq:ts_gen_def}).
\end{lemma}
This Lemma was proved in \cite{Couillet2009_CapacityAnalysisMIMO} for the case $s\in \RR^+$. We will assume it is valid for  $\Re(s)>-\lambda_{min}^*$. We should mention that for the case $K=1$, or for the case of equal correlation matrices $\vec R_k=\vec R_0$ ($k=1,\ldots,K$), the generalization to $\Re(s)<0$ can be inferred from \cite{Paul2009_NoEigenvaluesOutsideSupport}. From the above result and using the linearity of the derivative operation, we can deduce the ``hardening''of all derivatives of $\delta I(s,\vec H_0)$ with respect to $s$.
\begin{corollary}{(Hardening of Derivatives of $\delta I(s,\vec H_0)$)} \label{lemma:HardeiningDerivativesDI}
In the limit $M\rightarrow \infty$ the derivatives of $\delta I(s,\vec H_0)$ with respect to $s$ converge with high probability to their deterministic equivalents, which are the corresponding derivatives of $I_{erg}(s)$, defined in (\ref{eq:Ierg_gen_def}), (\ref{eq:rs_gen_def}), (\ref{eq:ts_gen_def}).
\end{corollary}
From the convexity of the function $\Delta I(s,\vec H_0)$ with respect to $s$, we can deduce that $I''_{erg}(s)<0$.

To show Proposition \ref{lemma:pdf}, we start by expressing the probability density function of $\gamma$ as follows
\begin{eqnarray}\label{eq:app_PDF_def1}
    P_M(\gamma) &=& M\int_{-i\infty}^{i\infty} \frac{ds}{2\pi i} e^{Ms\gamma} g_M(s) \\
    &=& \ex_{\vec H_0}\left[M\int_{-i\infty}^{i\infty} \frac{ds}{2\pi i} e^{Mf(s)}\right]
    %\\ &\equiv& \ex_{\vec H_0} F(\vec H_0)
\end{eqnarray}
where
\begin{equation}\label{eq:f(s)_def}
    f(s, \vec H_0)=s\gamma-\delta I(s,\vec H_0)
\end{equation}
Keeping in mind that in the large $M$ limit $\delta I(s,\vec H_0)=O(1)$ we proceed to {\em first} integrate over $s$ before averaging over $\vec H_0$. Since for $\Re(s)>-\lambda_{min}(\vec H_0)$ $f(s)$ is analytic, we deform the contour of the $s$-integral to pass through the saddle point(s) of $f(s)$ from the steepest descent path \cite{Bender_Orszag_book}, which are defined by $f'(s_0)=0$ or
\begin{eqnarray}
\gamma &=& \delta I'(s_0) \nonumber \\
&=& \frac{p_0}{M}\tr \left[\vec R_0 \left(\vec I_M + p_0 s_0\vec R_0+\vec H_0 \vec H_0^\dagger z^{-1}\right)^{-1}\right]
\label{eq:app_saddle_point_s_eq}
\end{eqnarray}
It is easy to see that the above equation only has real solutions for $Re(s_0)>-\lambda_{min}(\vec H_0)$. This is so, because in this region the right-hand-side is real only if $s_0$ is real. Also, since the right-hand-side above is a decreasing function of $s$, (becoming unbounded when $s\rightarrow -\lambda_{min}(\vec H_0)$ and going to zero when $s\rightarrow \infty$), it can also be shown that it can only have one solution, which depends on $\gamma$. Hence for $M\rightarrow \infty$ we expect the resulting limiting equation $\gamma = \delta I'_{erg}(s_0)$ to have a single real solution for $s>-\lambda_{min}^*$.

For large $M$, the integral will dominated by the behavior close to the saddle point. As a result, we may expand the exponent close to $s_0$. Thus
\begin{eqnarray}
f(s) &=& f(s_0) + \sum_{k=2}^\infty (-1)^k (s-s_0)^kf_k
\\ \nonumber
f_k &=& \frac{1}{M}\tr \left[\left(p_0\vec R_0 \left(\vec I_M + p_0 s_0\vec R_0+\vec H_0 \vec H_0^\dagger z^{-1}\right)^{-1}\right)^k\right]
\label{eq:f_k}
\end{eqnarray}
Since $f_2<0$ the steepest descent path in the neighborhood of $s_0$ is $s=s_0+it$, $t\in \RR$.  Keeping the first non-trivial term in the expansion of $f(s)$ (\ref{eq:f_k}) in the exponent, we expand the rest obtaining an expansion of the form
\begin{equation}\label{eq:expansion}
    \frac{M}{2\pi} e^{Mf(s_0)} \int_{-\infty}^\infty dt e^{-\frac{M}{2}|f_2|t^2} \left(1+\sum_{q=1}^\infty M^q A_q(t)\right)
\end{equation}
where the function $A_q(t)$ can be expressed as an expansion of $t$, with the minimum degree $3q$ if $q$ is even and minimum degree $3q+1$ if $q$ is odd. Integrating over $t$ and performing simple power counting of $M$ we conclude that to leading order in $M$ we have
\begin{eqnarray}\label{eq:app_PDF2}
    P(\gamma) &=& \sqrt{\frac{M}{2\pi}}\ex_{\vec H_0}\left[\frac{e^{M(s_0\gamma-\delta I(s_0,\vec H_0))}}{\sqrt{\left|\delta I''(s_0,\vec H_0)\right|}} \left(1+O(\frac{1}{M})\right)\right]
\end{eqnarray}
A number of comments are due for this expression. First, at least in principle, $\delta I(s_0,\vec H_0)$ and all its derivatives (given by $f_k$) are functions of the realization of $\vec H_0$, directly or through $s_0$, which is the solution of (\ref{eq:gamma_solution}). Nevertheless, from Corollary \ref{lemma:HardeiningDerivativesDI} we can replace the derivatives of $\delta I(s_0,\vec H_0)$ with their deterministic equivalents to leading order. As a result, to leading order in $M$ we have
\begin{eqnarray}\label{eq:app_PDF3}
    \overline{P}_M(\gamma) &=& \frac{\sqrt{M}e^{s_0\gamma}}{ {\sqrt{2\pi \left|\ex\left[\delta I''(s_0)\right]\right|}}} \ex_{\vec H_0}\left[e^{-M\delta I(s_0,\vec H_0)}\right] \\ \nonumber
    &=& \frac{\sqrt{M}e^{s_0\gamma}}{ {\sqrt{2\pi \left|\ex\left[\delta I''(s_0)\right]\right|}}} g(s_0) \\ \nonumber
\end{eqnarray}
To conclude the calculation, we need an expression of $g(s_0)$. Clearly, the ``hardening'' of the mutual information itself $\delta I(s_0,\vec H_0)$ has also been shown elsewhere. However, here we need an expression accurate to $O(1/M)$, hence we will need the next, i.e. $O(1)$ correction. This correction can be evaluated using the fact that $\Delta I(s_0,\vec H_0)$ is a difference of two MIMO mutual information functions with noise covariance matrix that differs by $s_0p_0\vec R_0$. We can then take advantage of a number of works in the literature that has analyzed the statistics of mutual information functions.
\begin{lemma}[CLT for $\Delta I(s,\vec H_0)$]\label{lemma:DI_Gaussian}
In the limit $M\rightarrow +\infty$, $N_k\rightarrow +\infty$ (for $k=1,\ldots,K$), such that $n_k=N_k/M$ remains finite, and for $s\geq -\lambda_{min}^*$ the quantity $\Delta I(s,\vec H_0)$ in (\ref{eq:DeltaI_def}) becomes asymptotically normal. In particular,
\begin{equation}\label{eq:Gaussian_DI_lemma}
    \frac{\Delta I(s,\vec H_0)-\Delta I_{erg}(s)}{\sqrt{v_1(s)}} \xrightarrow{M\to \infty} {\cal N}(0,1)
\end{equation}
where $\Delta I_{erg}(s_0)$ and its related parameters are given by (\ref{eq:Ierg_gen_def}). The variance $v_1(s)$ of $\Delta I_{erg}(s)$ is given by
\begin{eqnarray}\label{eq:app_variance}
    v_1(s) &=&  \var(I(s,\vec H_0)) + \var(I(0,\vec H_0)) \\ \nonumber
    &-& 2\cov(I(s,\vec H_k),I(0,\vec H_0))    \\ \nonumber
    &=& -\log\det\left|\vec I_{K}-\vec \Pi_{2}\vec \Sigma_{2}\right| \\ \nonumber
    &-& \log\det\left|\vec I_{K}-\vec \Pi_{0}\vec \Sigma_{0}\right| \\ \nonumber
    &+& 2\log\det\left|\vec I_{K}-\vec \Pi_{1}\vec \Sigma_{1}\right|
\end{eqnarray}
The elements of the positive-definite matrices $\vec \Pi$ and $\vec \Sigma$ are given below
\begin{eqnarray}\label{eq:app_SigmaPi_def}
    \Sigma_{2,ab} &=&  \delta_{ab} \left(\frac{p_a}{z+p_a t_a(s_0)}\right)^2 \\ \nonumber
    \Pi_{2,ab} &=& \frac{1}{N_a}\tr\left[\vec R_a\vec Q(s_0)^{-1}\vec R_b \vec Q(s_0)^{-1}\right] \\ \nonumber
    \Sigma_{0,ab} &=&  \delta_{ab} \left(\frac{p_a}{z+p_a t_a(0)}\right)^2 \\ \nonumber
    \Pi_{0,ab} &=& \frac{1}{N_a}\tr\left[\vec R_a\vec Q(0)^{-1}\vec R_b \vec Q(0)^{-1}\right] \\ \nonumber
    \Sigma_{1,ab} &=&  \delta_{ab} \frac{p_a^2}{(z+p_a t_a(s_0))(1+p_a t_a(0))} \\ \nonumber
    \Pi_{1,ab} &=& \frac{1}{N_a}\tr\left[\vec R_a\vec Q(s_0)^{-1}\vec R_b \vec Q(0)^{-1}\right] \\ \nonumber
\end{eqnarray}
for $a,b=1,\ldots,K$ and $z=1$ ($z=0$) for the MMSE (ZF) cases. The matrix $\vec Q(s)$ is defined as
\begin{eqnarray}\label{eq:app_Q_def}
    \vec Q(s) &=&  \vec I_M + p_0s\vec R_0+\sum_k \vec R_k r_k(s)
\end{eqnarray}
For convenience we generalize the above notation to include $\Pi_{2,ab}$, when any of its indices $a,b$ can take the value $0$, in which case the corresponding matrix $\vec R_a$ (and/or $\vec R_b$) becomes $\vec R_0$ and $N_{a=0}\rightarrow M$.
\end{lemma}

Although we do not formally prove this lemma, we will briefly motivate its validity and discuss how one can go about to prove it. The Gaussian behavior of MIMO mutual information functions was first introduced in \cite{Moustakas2003_MIMO1}, where in addition to the ergodic mutual information of the form appearing here, the variance of the difference of two mutual informations in both of which the same random matrix appears was calculated using the replica trick with both complex and Grasmann variables. Using this methodology the variance $v_1(s_0)$ above was evaluated. Furthermore it was shown that all higher cumulant moments vanish as increasing inverse powers of $M^{-1}$. This shows that $\Delta I(s,\vec H_0)$ converges to a Gaussian variable in the large $M$ limit. Similar results have been shown using more formal arguments for the case of a single mutual information function with Kronecker-correlated Gaussian channels in \cite{Hachem2006_GaussianCapacityKroneckerProduct} or with independent but not identically distributed channels \cite{Hachem2008_Capacity_INDCorrelation}.

Armed with the above result, we can now integrate over the channel $\vec H_0$ by changing variables, from $\vec H_0$ to the random Gaussian variable $Z=(\Delta I(s_0,\vec H_0)-\Delta I_{erg}(s_0))/\sqrt{v_1(s_0)}$. The reason we shift from $\Delta I(s_0,\vec H_0)$ to $Z$ is because we know from the analysis above that it is $Z$ that becomes asymptotically Gaussian. It is also the case that $g_M(s)$ itself involves the expectation of an exponentially small quantity ($e^{-\Delta I(s_0,\vec H_0)}$) when $M$ is large, hence its average is not necessarily well defined\footnote{For example, even if $Z$ above is asymptotically Gaussian the expectation $\ex\left[e^{-MZ}\right]$ is not well defined.}
\begin{eqnarray}\label{eq:app_PDF4}
    g(s_0)e^{\Delta I_{erg}(s_0)} &=& \ex_{\vec H_0}\left[e^{-\left(\Delta I(s_0,\vec H_0)-\Delta I_{erg}(s_0)\right)}\right] \\ \nonumber
    &=& \ex_{Z}\left[e^{-Zv_1(s_0)}\right] \\ \nonumber
    &=& \int_{-\infty}^{\infty} \frac{dz}{\sqrt{2\pi}} e^{-\frac{z^2}{2}-z v_1(s_0)}\left(1+O(M^{-1})\right) \\ \nonumber
    &=& e^{\frac{v_1(s_0)}{2}} (1+O(M^{-1}))
\end{eqnarray}
The corrections of order $O(1/M)$ stem from a number of sources. Specifically, the correction to $\Delta I(s_0,\vec H_0)-\Delta I_{erg}(s_0)$ is $O(1/M)$ \cite{Moustakas2003_MIMO1}, while the correction to the variance $v_1(s_0)$ is $O(1/M^2)$ \cite{Moustakas2003_MIMO1}. Both these corrections result to an $O(1/M)$ correction to the above result. Also, for finite large $M$ we may incorporate corrections to the Gaussian approximation by including the higher order statistics, e.g. the skewness \cite{Bouchaud_book_FinancialRiskDerivativePricing}. Here again the leading contribution stems from the skewness, which is $O(1/M)$ \cite{Moustakas2003_MIMO1}.

The expressions in Lemma \ref{lemma:DI_Gaussian} allow us to express the second derivative of $\delta I_{erg}(s_0)$ with respect to $s$ as follows:
\begin{eqnarray}\label{eq:app_saddle_pt_second_deriv}
\delta I''_{erg}(s_0)&=& -p_0\sum_{j=1}^{K} \Pi_{2,0j} \frac{dr_j(s_0)}{ds} -p_0^2\Pi_{2,00}\\ \nonumber
&=&-p^2_0\sum_{k,j=1}^{K} \Pi_{2,0j} \left(\vec \Sigma_2\left[\vec I_{K} - \vec \Pi_2 \vec\Sigma_2\right]^{-1}\right)_{jk} \Pi_{2,0k} \\ \nonumber
&& -p_0^2 \Pi_{2,00} \\ \nonumber
&\equiv& -e^{-v_2(s_0)}
\end{eqnarray}
The second equality follows from the expression of the derivatives of $r_k(s)$, $t_k(s)$ in (\ref{eq:rs_gen_def}), (\ref{eq:ts_gen_def}) with respect to $s$ in terms of the matrices $\vec \Pi_2$, $\vec \Sigma_2$. Finally, the last line above defines $v_2(s)$ in (\ref{eq:PDF_sinr1}).

%\IEEEproof{}

\begin{figure*}%[htbp]
\begin{center}
\subfigure[BER for MMSE SINR]
{\includegraphics[width=0.49 \textwidth]{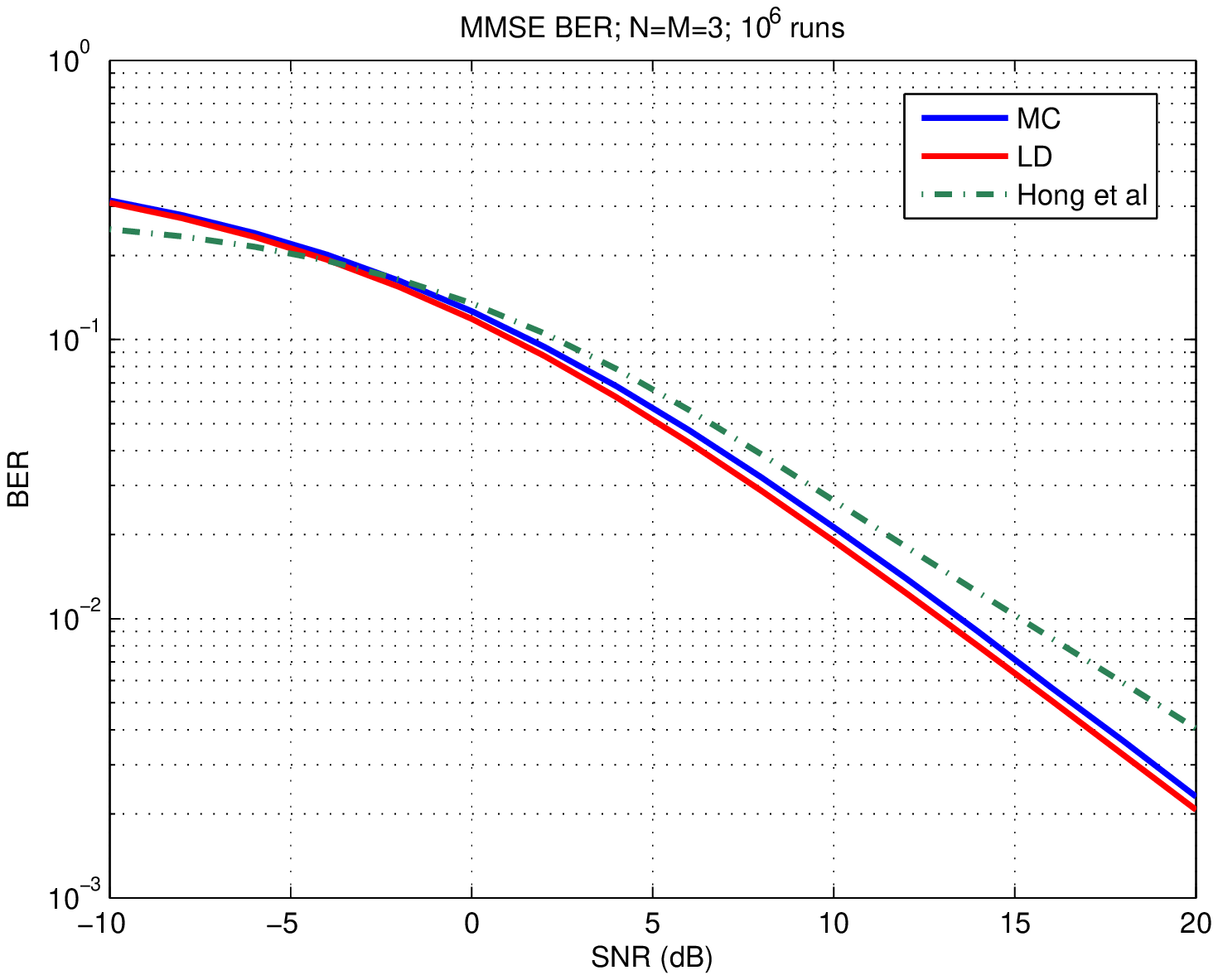}}
\subfigure[BER for ZF SINR]
{\includegraphics[width=0.49 \textwidth]{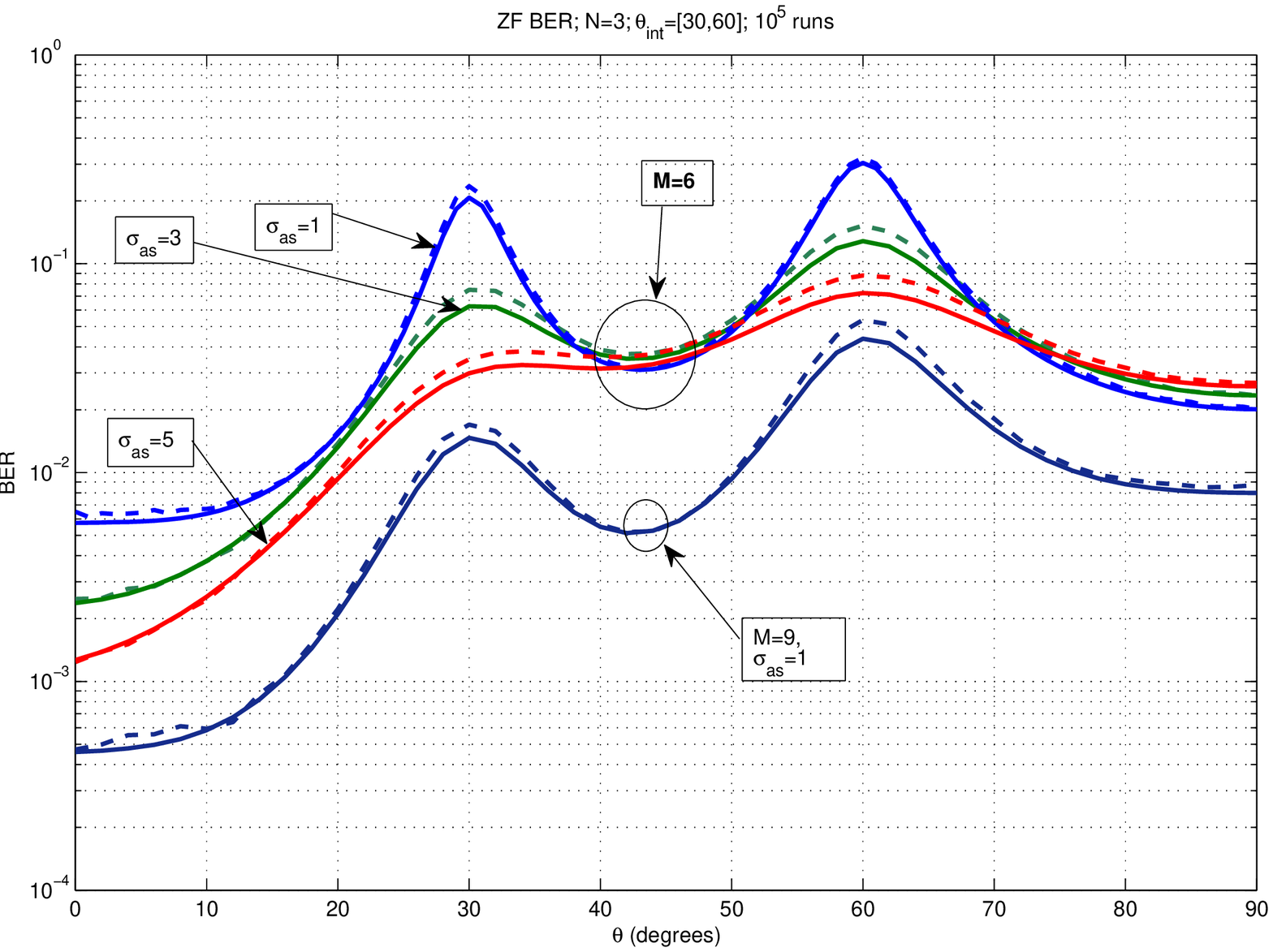}}
\caption{Uncoded BER for (a) MMSE and (b) ZF SINR. In the MMSE case, we take $M=N=3$ and compare to Monte Carlo simulations and generalized gamma distributions with iid channels. In the ZF case, we take $N=3$ and $M=6$ or $M=9$. We plot the BER as a function of the angle of the signal path in the presence of two interfering paths arriving at $\theta=30,60^o$, with various angle spreads $\sigma_{as}=1^o, 3^o,5^o$. The dashed lines in (b) are the Monte Carlo simulation results.}
\end{center}
\label{fig:BER_curves}
\end{figure*}

\section{Proof of Proposition \ref{lemma:cdf}}
\label{app:proof_cdf}

We will now provide some details in the proof of (\ref{eq:CDF_sinr1}). We will deal only with the case $\gamma<\gamma_{erg}$, since the opposite case $\gamma>\gamma_{erg}$ can be analyzed in a similar way. $P_{M,out}(\gamma)$ is defined as
\begin{eqnarray}
\label{eq:CDF_proof1}
  P_{M,out}(\gamma) &=& \int_0^\gamma dt \overline{P}_M(t)
\end{eqnarray}
up to negligible corrections $O(M^{-1})$ due to replacing $\overline{P}_M(t)$ for $P_M(t)$. The analysis is based on the fact that for large $M$ the outage probability $P_{M,out}(\gamma)$ is determined from the behavior of $P_M(t)$ close to the We will need to focus separately in two regions of interest in the interval $\gamma\in(0,\gamma_{erg}]$. In the first region $|\gamma-\gamma_{erg}|=O(1)$, to asymptotically evaluate the outage probability we expand the exponent of $\overline{P}_M(t)$ (\ref{eq:PDF_sinr1}) in $\gamma$ around the end point of the integral. Since $\overline{P}_M(\gamma)$ is an increasing function for $\gamma<\gamma_{erg}$ its derivative will be always positive in this region. Hence we have
\begin{eqnarray}
\label{eq:CDF_proof2}
  P_{M,out}(\gamma) &\approx&  \int_0^\gamma \frac{dt}{\sqrt{2\pi}}  e^{Ms_0(t-\gamma)} \\ \nonumber
  && \cdot e^{Ms_0\gamma-\Delta I_{erg}(s_0)+\frac{v_1(s_0)+v_2(s_0)}{2}} \\ \nonumber
   &=&  \frac{e^{Ms_0\gamma-\Delta I_{erg}(s_0)+\frac{v_1(s_0)+v_2(s_0)}{2}} }{\sqrt{2\pi }Ms_0} \left(1+O(M^{-1})\right)
\end{eqnarray}
where $s_0$ above is evaluated at the endpoint $\gamma$. We have used the fact that to leading the derivative of the exponent {\em with respect to $\gamma$} is simply $Ms_0$. The above approximation begins to break down when $|s_0|\sqrt{M}\ll 1$, i.e. in the region $\sqrt{M}|\gamma-\gamma_{erg}|=O(1)$. Although this situation will rarely occur when we take the limit $M\rightarrow \infty$ for {\em fixed $0<\gamma<\gamma_{erg}$} it useful to pay attention to this region so that we can provide an approximation that is valid for every $\gamma$ when $M$ is large but fixed. In this increasingly diminishing region as $M\rightarrow \infty$, the Gaussian approximation of the SINR is  valid, where the probability density of $P_M(\gamma)$ will be  approximately quadratic in $\gamma$.  Hence we expand the exponent of $\overline{P}_M(t)$ to second order around the endpoint $t=\gamma$, and then integrate over $t$. After some algebra and using the fact that
\begin{equation}\label{eq:s_prime_def}
    \frac{ds_0(\gamma)}{d\gamma}=\frac{1}{\delta I_{erg}''(s_0(\gamma))}
\end{equation}
we obtain  (\ref{eq:CDF_sinr1}). To obtain the expression in (\ref{eq:CDF_sinr2}) we express $P_{out}(\gamma_0)=1-\prob(\gamma>\gamma_0)$ and work as above with $\prob(\gamma>\gamma_0)$. The final expressions (\ref{eq:CDF_sinr1}), (\ref{eq:CDF_sinr2}) smoothly interpolate between (\ref{eq:CDF_proof2}) (for $|\gamma-\gamma_{erg}|=O(1)$) and the Gaussian approximation (for $|\gamma-\gamma_{erg}|\sqrt{M}=O(1)$).

\section{Proof of Proposition \ref{lemma:ber}}
\label{app:proof_ber}

The average uncoded bit-error rate (BER) for signals with modulation as in (\ref{eq:MQAM_def}) can be expressed in terms of the moment-generating function $g(is)$ as follows
\begin{eqnarray}\label{eq:app_proof_ber1}
BER_M &=& b\int_0^\infty d\gamma P_M(\gamma) Q(\sqrt{a\gamma}) \\ \nonumber
 &=& \frac{b}{2}\left(1-\int_{-\infty}^{+\infty} \frac{ds}{2\pi i} \frac{g_M(is)}{s+i0^+}\frac{1}{\sqrt{1-i2sM/a}}\right) \\ \nonumber
 &=& \frac{b}{2\pi}\int_{0}^{\infty}  \frac{dt}{\sqrt{t}} \frac{g_M\left(\frac{a(1+t)}{2M}\right)}{t+1}
\end{eqnarray}
In the first equation the parameters $a$, $b$ correspond to the different cases in (\ref{eq:MQAM_def}). The second equation results from the definition of $P_M(\gamma)$ in terms of the moment-generating function. The third equation follows by deforming the integral from the real axis to follow the branch cut appearing due to the square root. Using (\ref{eq:app_PDF4}) to express $g_M(s)$ in terms of $\Delta I_{erg}(s)$ etc, we get
\begin{eqnarray}\label{eq:app_proof_ber2}
BER &=& \frac{b}{2\pi}\int_0^\infty \frac{dt}{\sqrt{t}(1+t)} e^{-\Delta I_{erg}\left(\frac{a(1+t)}{2M}\right)+\frac{1}{2}v_1\left(\frac{a(1+t)}{2M}\right)} \nonumber \\
&=& \frac{b}{2\pi}e^{-\Delta I_{erg}\left(\frac{a}{2M}\right)+\frac{1}{2}v_1\left(\frac{a}{2M}\right)} \\ \nonumber
&&\cdot \int_0^\infty \frac{dt}{\sqrt{t}(1+t)} e^{-\frac{a}{2M}I'_{erg}\left(\frac{a}{2M}\right)t}(1+O(M^{-1}))
\end{eqnarray}
In the second line above we have expanded the exponent for small arguments and kept only the $O(M)$ and $O(1)$, neglecting all lower order terms. Integrating the above expression over $t$ gives (\ref{eq:lemma_ber1}). It should be noted that if we wanted to be strict regarding the leading corrections being $O(M^{-1})$, in the above expression the arguments of $v_1(s)$, $I_{erg}(s)$ and $I'_{erg}(s)$ should be set to $s=0$, rather than $s=a/(2M)$. Nevertheless, we have found numerically that these expressions are slightly more accurate.


%\pagebreak


\begin{thebibliography}{10}
\providecommand{\url}[1]{#1}
\csname url@samestyle\endcsname
\providecommand{\newblock}{\relax}
\providecommand{\bibinfo}[2]{#2}
\providecommand{\BIBentrySTDinterwordspacing}{\spaceskip=0pt\relax}
\providecommand{\BIBentryALTinterwordstretchfactor}{4}
\providecommand{\BIBentryALTinterwordspacing}{\spaceskip=\fontdimen2\font plus
\BIBentryALTinterwordstretchfactor\fontdimen3\font minus
  \fontdimen4\font\relax}
\providecommand{\BIBforeignlanguage}[2]{{%
\expandafter\ifx\csname l@#1\endcsname\relax
\typeout{** WARNING: IEEEtran.bst: No hyphenation pattern has been}%
\typeout{** loaded for the language `#1'. Using the pattern for}%
\typeout{** the default language instead.}%
\else
\language=\csname l@#1\endcsname
\fi
#2}}
\providecommand{\BIBdecl}{\relax}
\BIBdecl
\renewcommand{\BIBentryALTinterwordstretchfactor}{4}

\bibitem{Foschini1998_BLAST1}
G.~J. Foschini and M.~J. Gans, ``On limits of wireless communications in a
  fading environment when using multiple antennas,'' \emph{Wireless Personal
  Communications}, vol.~6, pp. 311--335, 1998.

\bibitem{Telatar1995_BLAST1}
I.~E. Telatar, ``Capacity of multi-antenna {G}aussian channels,''
  \emph{European Transactions on Telecommunications and Related Technologies},
  vol.~10, no.~6, pp. 585--596, Nov. 1999.

\bibitem{Standard802.11n}
``\mbox{Draft Standardization Document, IEEE P802.11n/D2.00},'' Tech. Rep.,
  Feb. 2007.

\bibitem{Muller2001_LowComplexityLinearReceivers}
R.~R. M\"uller and S.~Verd\'u, ``Design and analysis of low-complexity
  interference mitigation on vector channels,'' \emph{{IEEE} J. Select. Areas
  Commun.}, vol.~19, no.~8, p. 1429, Aug. 2001.

\bibitem{Wolniansky1998_VBLAST_Indoors1}
P.~W. Wolniansky \emph{et~al.}, ``\mbox{V-BLAST}: an architecture for realizing
  very high data rates over the rich-scattering wireless channel,'' \emph{URSI
  International Symposium on Signals, Systems and Electronics}, pp. 295--300,
  1998.

\bibitem{Caire2003_LowComplexitySTCoding}
G.~Caire and G.~Colavolpe, ``On low-complexity space-time coding for
  quasi-static channels,'' \emph{IEEE Trans. Inf. Theory}, vol.~49, no.~6, pp.
  1400--1416, Jun. 2003.

\bibitem{Tse1999_LinearMultiuserReceiversEffectiveInterferenceEtc}
D.~N. Tse and S.~V. Hanly, ``Linear multiuser receivers: Effective
  interference, effective bandwidth and user capacity,'' \emph{{IEEE} Trans.
  Inform. Theory}, vol.~45, no.~2, p. 641, Mar. 1999.

\bibitem{Verdu1999_MIMO1}
S.~Verd{\'u} and S.~Shamai, ``Spectral efficiency of \mbox{CDMA} with random
  spreading,'' \emph{{IEEE} Trans. Inform. Theory}, vol.~45, no.~2, pp.
  622--640, Mar. 1999.

\bibitem{Peacock2006_UnifiedLinearReceiversRMT}
M.~J.~M. Peacock, I.~B. Collings, and M.~L. Honig, ``Unified large-system
  analysis of {MMSE} and adaptive least squares receivers for a class of random
  matrix channels,'' \emph{{IEEE} Trans. Inform. Theory}, vol.~52, no.~8, p.
  3567, Aug 2006.

\bibitem{Chaufray2008_CDMA_MMSE_RMT}
J.~Chaufray, W.~Hachem, and P.~Loubaton, ``Asymptotic analysis of optimum and
  sub-optimum {CDMA} downlink {MMSE} receivers,'' \emph{{IEEE} Trans. Inform.
  Theory}, vol.~50, no.~11, pp. 2620--2638, Nov. 2004.

\bibitem{Artigue2010_PrecoderDesignMIMO_MMSE_RMT}
C.~Artigue and P.~Loubaton, ``On the precoder design of flat fading {MIMO}
  systems equipped with {MMSE} receivers: {A} large system approach,''
  \emph{accepted in IEEE Trans. on Inform. Theory}, 2010, available at
  http://arxiv.org/abs/0911.0351.

\bibitem{McKay2010_MI_MMSE_RMT}
M.~R. McKay, I.~B. Collings, and A.~M. Tulino, ``Achievable sum rate of {MIMO}
  {MMSE} receivers: {A} general analytic framework,'' \emph{{IEEE} Trans.
  Inform. Theory}, vol.~56, no.~1, pp. 396--410, Jan. 2010.

\bibitem{Kumar2009_LinearDMT}
K.~R. Kumar, G.~Caire, and A.~L. Moustakas, ``Asymptotic performance of linear
  receivers in {MIMO} fading channels,'' \emph{{IEEE} Trans. Inform. Theory},
  vol.~55, no.~10, p. 4398, Oct. 2009.

\bibitem{Biglieri1998_FadingChannels}
E.~Biglieri, J.~Proakis, and S.~Shamai, ``Fading channels:
  Information-theoretic and communications aspects,'' \emph{{IEEE} Trans.
  Inform. Theory}, vol.~44, no.~6, p. 2619, Oct. 1998.

\bibitem{Tse2000_MMSEFluctuations}
D.~N. Tse and O.~Zeitouni, ``Linear multiuser receivers in random
  environments,'' \emph{{IEEE} Trans. Inform. Theory}, vol.~46, no.~1, p. 171,
  Jan. 2000.

\bibitem{Zhang2001_Gaussian_forMAI_MMSE_DS_CDMA}
J.~Zhang, E.~K.~P. Chong, and D.~N.~C. Tse, ``Output {MAI} distribution of
  linear {MMSE} multiuser receivers in {DS-CDMA} systems,'' \emph{{IEEE} Trans.
  Inform. Theory}, vol.~47, no.~3, pp. 1128--1144, Mar. 2001.

\bibitem{Guo2002_AsymptoticNormalityMMMSE_ZF}
D.~Guo, S.~Verdu, and L.~K. Rasmussen, ``Asymptotic normality of linear
  multiuser receiver outputs,'' \emph{{IEEE} Trans. Inform. Theory}, vol.~48,
  no.~12, pp. 3080 -- 3095, Dec. 2002.

\bibitem{Liang2007_MMSEAsymptotics}
Y.-C. Liang, G.~Pan, and Z.~D. Bai, ``Asymptotic performance of mmse receivers
  for large systems using random matrix theory,'' \emph{{IEEE} Trans. Inform.
  Theory}, vol.~53, no.~11, p. 4173, Nov. 2007.

\bibitem{Kammoun2009_CLT_MMSE_RMT}
A.~Kammoun \emph{et~al.}, ``A central limit theorem for the {SINR} at the
  {LMMSE} estimator output for large dimensional systems,'' \emph{{IEEE} Trans.
  Inform. Theory}, vol.~55, no.~11, pp. 5048--5063, Nov 2009.

\bibitem{Gore2002_MIMO_ZFReceiver}
D.~Gore, R.~W. Heath, and A.~Paulraj, ``On performance of the zero forcing
  receiver in presence of transmit correlation,'' in \emph{Proc. IEEE Int.
  Symp. Information Theory}, Lausanne, Switzerland, Jun 2002, p. 159.

\bibitem{Muller1997_CDMA_CapacityLinearReceivers}
R.~R. M\"uller, P.~Schramm, and J.~B. Huber, ``Spectral efficiency of {CDMA}
  systems with linear interference suppression,'' in \emph{IEEE Workshop on
  Communication Engineering}, Ulm, Germany, Jan 1997, pp. 93--97.

\bibitem{Schramm1999_SpectralEfficiencyMMSEReceivers}
P.~Schramm and R.~R. M\"uller, ``Spectral efficiency of {CDMA} systems with
  linear {MMSE} interference suppression,'' \emph{{IEEE} Trans. Commun.},
  vol.~47, no.~5, pp. 722 --731, May 1999.

\bibitem{Li2006_MIMO_MMSE_SINR_Distribution}
P.~Li \emph{et~al.}, ``On the distribution of {SINR} for the {MMSE} {MIMO}
  receiver and performance analysis,'' \emph{{IEEE} Trans. Inform. Theory},
  vol.~52, no.~1, p. 271, Jan. 2006.

\bibitem{Kammoun2009_BER_Outage_Approximations_MMSE_MIMO}
A.~Kammoun \emph{et~al.}, ``{BER} and outage probability approximations for
  {LMMSE} detectors on correlated {MIMO} channels,'' \emph{{IEEE} Trans.
  Inform. Theory}, vol.~55, no.~10, pp. 4386--4397, Oct 2009.

\bibitem{Armada2009_BitLoadingMIMO}
A.~G. Armada, L.~Hong, and A.~Lozano, ``Bit loading for {MIMO} with statistical
  channel information at the transmitter and {ZF} receivers,'' in
  \emph{Proceedings of the 2009 IEEE international conference on
  Communications}, ser. ICC'09.\hskip 1em plus 0.5em minus 0.4em\relax
  Piscataway, NJ, USA: IEEE Press, 2009, pp. 3836--3840.

\bibitem{Li2011_BER_MIMO_MMSE}
H.~Li and A.~G. Armada, ``Bit error rate performance of {MIMO} {MMSE} receivers
  in correlated {R}ayleigh flat-fading channels,'' \emph{{IEEE} Trans. Veh.
  Technol.}, vol.~60, no.~1, pp. 313 --317, Jan. 2011.

\bibitem{Kiessling2003_ExactMMSE_SINR}
M.~Kiessling and J.~Speidel, ``Analytical performance of {MIMO} {MMSE}
  receivers in correlated rayleigh fading environments,'' in \emph{Vehicular
  Technology Conference, 2003. VTC 2003-Fall. 2003 IEEE 58th}, vol.~3, Oct.
  2003, pp. 1738 -- 1742.

\bibitem{Moustakas2011_SINR_MMSE_PhysPolonicaB}
A.~L. Moustakas, ``Tails of composite random matrix diagonals: {T}he case of
  the {W}ishart inverse,'' \emph{Acta Phys. Polonica B, in press}, 2011,
  arXiv:1104.1910.

\bibitem{Majumdar2006_LesHouches}
S.~N. Majumdar, \emph{Random Matrices, the {U}lam Problem, Directed Polymers \&
  Growth Models, and Sequence Matching}, ser. Les Houches, M.~M{\'e}zard and
  J.~P. Bouchaud, Eds.\hskip 1em plus 0.5em minus 0.4em\relax Elsevier, July
  2006, vol. Complex Systems.

\bibitem{Verdu_MUD_book}
S.~Verd\'{u}, \emph{Multiuser Detection}.\hskip 1em plus 0.5em minus
  0.4em\relax Cambridge, UK: Cambridge University Press, 2003.

\bibitem{Moustakas2003_MIMO1}
A.~L. Moustakas, S.~H. Simon, and A.~M. Sengupta, ``\mbox{MIMO} capacity
  through correlated channels in the presence of correlated interferers and
  noise: \mbox{A} (not so) large \mbox{N} analysis,'' \emph{{IEEE} Trans.
  Inform. Theory}, vol.~49, no.~10, pp. 2545--2561, Oct. 2003.

\bibitem{Hachem2006_GaussianCapacityKroneckerProduct}
W.~Hachem \emph{et~al.}, ``A new approach for capacity analysis of large
  dimensional multi-antenna channels,'' \emph{{IEEE} Trans. Inform. Theory},
  vol.~54, pp. 3987--4004, Sep. 2008.

\bibitem{Hachem2008_Capacity_INDCorrelation}
W.~Hachem, P.~Loubaton, and J.~Najim, ``A {CLT} for information-theoretic
  statistics of {G}ram random matrices with a given variance profile,''
  \emph{Annals of Applied Probality}, vol.~18, pp. 2071--2130, 2008.

\bibitem{Couillet2009_CapacityAnalysisMIMO}
R.~Couillet, M.~Debbah, and J.~W. Silverstein, ``A deterministic equivalent for
  the analysis of correlated {MIMO} multiple access channels,'' \emph{{IEEE}
  Trans. Inform. Theory}, vol.~57, no.~6, pp. 3493 --3514, June 2011.

\bibitem{Dupuy2010_CapacityAchievingCovarianceFrequencySelectiveMIMOChannels}
F.~Dupuy and P.~Loubaton, ``On the capacity achieving covariance matrix for
  frequency selective {MIMO} channels using the asymptotic approach,'' in
  \emph{Intern. Symp. on Inform. Theory}, Austin, TX, USA, June 2010, pp.
  2153--2157, under revision, IEEE Trans. Inform. Theory.

\bibitem{Moustakas2009_VTC_MMSE_CorrMIMO_Capacity}
A.~L. Moustakas, K.~R. Kumar, and G.~Caire, ``Performance of {MMSE} {MIMO}
  receivers: A large {N} analysis for correlated channels,'' in \emph{IEEE 69th
  Vehicular Technology Conference, Spring 2009}, 2009, pp. 1--5.

\bibitem{Moustakas2011_SINR_MMSE_ISIT}
A.~L. Moustakas, ``{SINR} distribution of {MIMO} {MMSE} receiver,'' in
  \emph{IEEE Intern. Symposium on Information Theory}, Aug 5 2011, pp. 938
  --942.

\bibitem{Moustakas2007_MIMO1}
A.~L. Moustakas and S.~H. Simon, ``On the outage capacity of correlated
  multiple-path {MIMO} channels,'' \emph{{IEEE} Trans. Inform. Theory},
  vol.~53, no.~11, pp. 3887--3903, Nov. 2007.

\bibitem{Paul2009_NoEigenvaluesOutsideSupport}
D.~Paul and J.~W. Silverstein, ``No eigenvalues outside the support of the
  limiting empirical spectral distribution of a separable covariance matrix,''
  \emph{J. Multivar. Anal.}, vol. 100, no.~1, pp. 37--57, Jan. 2009.

\bibitem{Bender_Orszag_book}
C.~M. Bender and S.~A. Orszag, \emph{Advanced Mathematical Methods for
  Scientists and Engineers}.\hskip 1em plus 0.5em minus 0.4em\relax New York,
  NY: McGraw-Hill, 1978.

\bibitem{Bouchaud_book_FinancialRiskDerivativePricing}
J.-P. Bouchaud and M.~Potters, \emph{Theory of Financial Risk and Derivative
  Pricing}, 2nd~ed.\hskip 1em plus 0.5em minus 0.4em\relax Cambridge, UK:
  Cambridge, 2003.

\end{thebibliography}
\end{document}